\documentclass[aps,prl,preprint,nofootinbib]{revtex4}
\usepackage{CJK}
\usepackage[cp1251]{inputenc}
\usepackage{amsmath,amsthm}
\usepackage{paralist}
\usepackage{graphicx}
\usepackage{epstopdf}
\usepackage[colorlinks=true]{hyperref}
\usepackage{float}
\usepackage{color,comment}
\usepackage{caption}
\usepackage{wrapfig}
\usepackage{multirow}
\usepackage{tabularx}
\usepackage{booktabs}
\usepackage{subfigure}
\usepackage[scientific-notation=true]{siunitx}
\DeclareMathOperator{\trace}{Tr}
\DeclareMathOperator{\deter}{Det}
\newtheorem{thm}{Theorem}
\newtheorem{remark}{Remark}

\newtheorem{lemma}{Lemma}
\usepackage[parfill]{parskip}
\makeatletter
\def\subsection{\@startsection{subsection}{2}%
	\z@{.5\linespacing\@plus.7\linespacing}{.1\linespacing}%
	{\normalfont\bfseries}}
\makeatother

\begin{document}
\begin{CJK*}{GB}{} 

\title[Language Dynamics]{The ``higher" status  language doesn't always win: The fall of English in India and the rise of Hindi}

\author{ Kushani De Silva}
\affiliation{Department of Mathematics, Iowa State University, Ames, IA 50011.}
\author{Aladeen Basheer}
\affiliation{Department of Mathematics, University of Georgia, Athens, GA 30602. }
\author{Kwadwo Antwi-Fordjour}
\affiliation{Department of Mathematics and Computer Science, Samford University, Birmingham, AL 35229.}
\author{Matthew A. Beauregard}
\affiliation{Department of Mathematics, Stephen F. Austin State University, Nacogdoches, TX 75962.}
\author{Vineeta Chand }
\noaffiliation
\author{Rana D. Parshad}
\affiliation{Department of Mathematics, Iowa State University, Ames, IA 50011.}

	\maketitle
%
%
%
%
%
\end{CJK*}


\begin{abstract}
Classical language dynamics explains language shift as a process in which speakers adopt a higher status language in lieu of a lower status language. This is well documented with English having out-competed languages such as Scottish Gaelic, Welsh and Mandarin.
The 1961-1991 Indian censuses report a sharp increase in Hindi/English Bilinguals, suggesting that English is on the rise in India - and is out-competing Hindi. However, the 1991 - 2011 data shows that Bilingual numbers have saturated, while Monolingual Hindi speakers continue to rise exponentially.
To capture this counter-intuitive dynamic, we propose a novel language dynamics model of interaction between Monolingual Hindi speakers and Hindi/English Bilinguals, which captures the Indian census data of the last 50 years with near perfect accuracy, outperforming the best known language dynamics models from the literature.
We thus provide a first example of a lower status language having out competed a higher status language.
\end{abstract}

\maketitle
\section{Introduction}

\subsection{Motivation}

Languages compete, just as species do, for speakers in a population \cite{Nie13, Pat09,bg12,pr06,Kand09, An12}. Methods from statistical physics, evolutionary biology, dynamical systems, game theory and agent-based modeling have been extremely effective in analysing language change and shift, see \cite{Nie13, L07,CFL09, L11, N99, muf08,Monojit,Prochazka,blythe} and the references within. Language shift towards English has also been under intense investigation \cite{Nie13, Kand09, W14, Stro03, Min08, Now02, Wang05, Nie15, IF14, Mir05, Mir11, Bag90, Hein14}. These models have successfully shown how English has out-competed Scottish Gaelic in Scotland, Welsh in Wales and Mandarin in Singapore \cite{Stro03,Gon13}.
The essential approach in the literature is to formulate the language dynamics problem as a two species competition problem, where the more prestigious language is the stronger competitor, and the less prestigious one the weaker competitor. Consider a population in which the speakers have a choice of either language A (the more prestigious one) or B (the less prestigious one).  Denote $u(t)$ the fraction of the population that speaks A, and
$v(t)$ the fraction of the population that speaks B, so $u+v=1$. In this setting we can write down a differential equation for the change in the fraction of the populations $u$,

\begin{equation}
\label{Eqn:lq}
  \frac{du}{dt} =  \overbrace{\alpha_{1}f(u)v}^{\mbox{ speakers switching to A}} -  \overbrace{\alpha_{2}h(u)v}^{\mbox{ speakers switching  to B}} .
\end{equation}

The functions $f$ and $h$ are typically of Lokta-Volterra type.
These models predict steady states of $(1,0)$ or $(0,1)$, depending on which language is more prestigious - but essentially the stronger language (competitor) wipes out the other.
Alternatively, three species models that incorporate bilingualism have also been considered \cite{Mir05, Min08}, and it is shown that under certain constraints on inter-linguistic similarity actual three language groups (2 Monolingual groups and a Bilingual group) can all co-exist \cite{Mir11}. However, models for language competition are not one size fits all. In the Indian context, exploring language competition between Hindi and English using earlier models with an assumption of one high prestige and the latter low prestige, the two languages recognized in the constitution for use across India, does not yield an ideal fit of census data from 1961-2011, as we will show. Thus a nuanced exploration of modeling language competition in evolving social contexts that take into account local ecological factors may be an alternative approach.

%
%

\subsection{Trends in the Indian Data}


In order to motivate our analysis we look at the Indian census data divided into two periods, 1961-1991 and 1991-2011.
During 1961-1991 the Bilingual English-Hindi population grew faster than the Monolingual Hindi population ($\approx e^{0.062 t}$ vs $\approx e^{0.025 t}$), see Fig. \ref{fig:ExpFitsub1}. However, during 1991-2011 the Bilingual English-Hindi population completely saturated - whereas the Monolingual Hindi population continues to grow at $\approx e^{0.0249 t}$, see Table \ref{table: exp fits} and Fig. \ref{fig:ExpFitsub2}. Interestingly, if one focuses only on the 1961-1991 data, classical language dynamics models \cite{IF14}, provide the best fit, see Table \ref{table: smt}.

\begin{table}[H]
		\caption{Estimated parameter values of the two language groups in India at two periods, 1961-1991 and 1991-2011. The format of the fitted functions are given in parenthesis.}
	\begin{tabular}{llll}
		\hline
		Period                     & Group                & Parameter Estimates                       & SSE                   \\\hline
		\multirow{2}{*}{1961-1991} & Monolingual ($a \exp (bt)$)         & $a=2.898 \times 10^{-23}, b=0.025$        & $3.025\times 10^{-4}$ \\
		& Bilingual  ($a \exp (bt)$)          & $a=8.855 \times 10^{-56}, b=0.062$        & $3.154\times 10^{-6}$ \\\hline
		\multirow{3}{*}{1991-2011} & Monolingual ($a \exp (bt)$)         & $a=8.234 \times 10^{-23}, b=0.0249$       & $7.738\times 10^{-6}$ \\
		& Bilingual ($p_1t + p_2$)  & $p_1 = 1.509 \times 10^{-4},p_2 = -0.271$ & $2.312\times 10^{-6}$ \\
		& Bilingual ($p_2$) & $p_2 = 0.031$                               & $6.916\times 10^{-6}$ \\ \hline
	\end{tabular}
\label{table: exp fits}
\end{table}

\begin{figure}[H]
	\centering
	\subfigure[caption]{	\includegraphics[width=.5\linewidth]{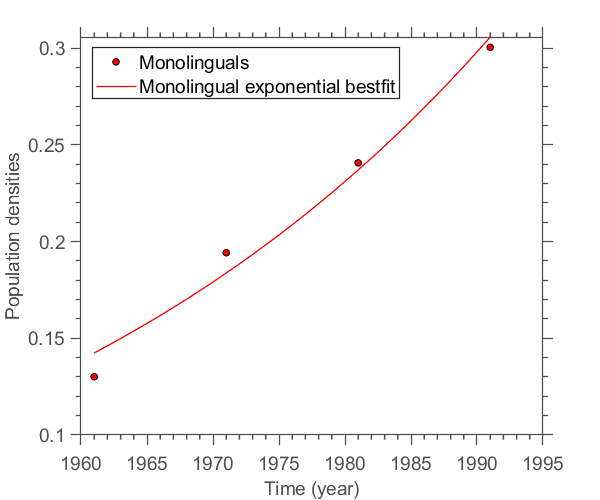}}\label{fig:ExpFit_Mono}%
\subfigure{\includegraphics[width=.5\linewidth]{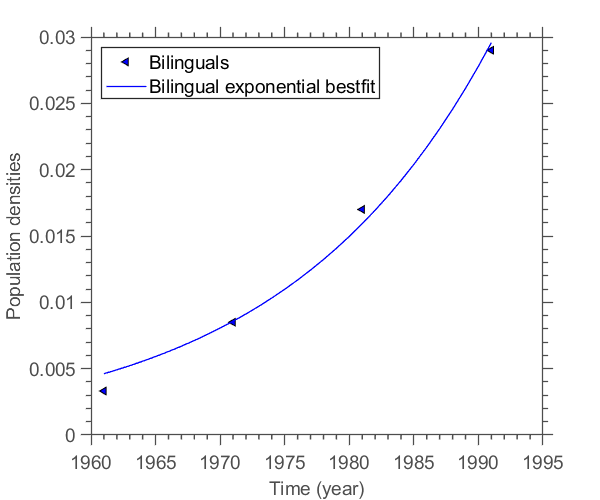}}\label{fig:ExpFit_Bilin}
	\caption{The exponential fit for Monolingual (left) and Bilingual (right) census data from 1961-1991.}
	\label{fig:ExpFitsub1}
\end{figure}

\begin{figure}[H]
	\centering
	\subfigure{\includegraphics[width=.5\linewidth]{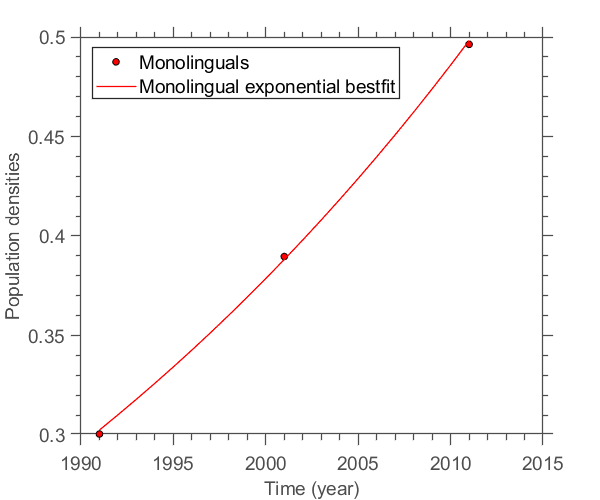}}\label{fig:expfit_mono_sub2}%
\subfigure{\includegraphics[width=.5\linewidth]{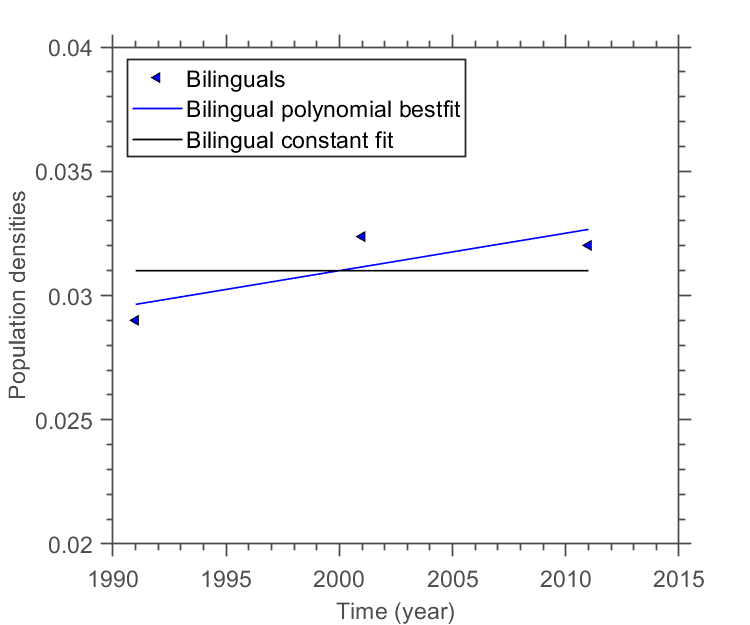}}\label{fig:expfit_bilin_sub2}
	\caption{The exponential fit for Monolingual (left) and linear fit for Bilingual (right) census data from 1991-2011.}
	\label{fig:ExpFitsub2}
\end{figure}

\subsection{Local Ecological Factors}

We focus on a narrower setting within India, the Hindi Belt, which is a swath across north/central India encompassing the capital, New Delhi, and including a majority of the Indian population, in which Hindi has a stronghold \cite{Jaffrelot2000, LaDousa2004}. More broadly, a large proportion of L1 Hindi speakers all over India, are from these states.
Note, Indian politics are inextricably tied to language.
Hindi as a signified product is linked in today's India to Hindutva - a nationalist sentiment \cite{Ninan2017} which imagines and promotes a monolithic homogeneous and hegemonic Hindu identity for all Indians. This right-wing nationalist ideology erases indigenous ethnic, religious and linguistic diversity, and challenges the secular, multicultural ideals upon which India was established.
While English media dominated the national marketplace from India's 1947 independence through the early 1980's, this is no longer true, and Hindi (along with other vernacular languages) now dominates Indian media: compare the Hindi market share of 47.7$\%$ to English, which accounts for only 11.4$\%$ of print daily newspapers \cite{Neyazi2018}. Specifically, there has been a rapid increase in Hindi-medium print and TV media since the early 1980's, especially in rural areas \cite{Neyazi2018}. 
Also, the path to English bilingualism, when coming from a monolingual Hindi background, is challenging in a number of ways. While international schools, one benchmark of English medium education in India, have increased, these schools exclude lower classes because of tuition costs, and exclude rural communities, given that they are situated in urban areas. Thus resources to learn English are getting limited. \footnote{Contemporary language hybridization is also visible in unmarked code-switching (alternating between two or more languages in a single conversation) between English and regional Indian vernaculars.
	In the Hindi Belt, Hinglish is the most prominent form of hybrid communication.  ``Hinglish" is a colloquial umbrella-term \cite{Kot11, Chau13, Si10} spanning isolated borrowings to rich code-switching practices unintelligible to Monolingual Hindi or English speakers.}
However, for census purposes Hinglish speakers, would list themselves as Bilinguals. Thus from a modeling standpoint (and in light of the 2011 census data), we divert from our three species framework proposed in \cite{PC16} and include Hinglish speakers in the Bilingual class.

\section{Model System}
We define,

\begin{itemize}
\item Monolingual Hindi class (M): Can produce Monolingual Hindi, English restricted to limited inclusion of established indigenizations and loanwords.

\item Hindi/English Bilingual class (B): Can produce Monolingual Hindi and Monolingual English.
This class also contains an urban sub-population that cannot produce pure Monolingual Hindi, and/or Monolingual English, only Hinglish - but have a certain degree of competency in both.
\end{itemize}
 We next describe our compartmental model, describing the interaction between $M(t)$ and $B(t)$, the populations of the Monolingual Hindi and Bilingual English/Hindi speaking communities,
\begin{eqnarray}
\label{ODE:M}
\dot{M} &=& \frac{a_1M}{1+d_1 B}  - \frac{a_{MB}MB}{1+d_2M} - b_1 M \equiv f_M(M,B)M\\
\dot{B} &=& \frac{a_{MB}MB}{1+d_2M} - b_2 B^2   \equiv f_B(M,B)B.\notag
\end{eqnarray}
The parameters are all assumed to be positive and their descriptions are given in Table \ref{tab:sym meaning}.
\begin{table}[H]
	\begin{center}
		\caption{List of parameters of the ODE system in Eq.~\eqref{ODE:M} and their contextual meanings.}
		\label{tab:sym meaning}
		\begin{tabular}{@{}l l@{}}
			\hline
			Symbol & Meaning \\
			\hline
			$a_1$                  & Growth rate of $M$ \\
			$b_1$                  & Natural mortality of $M$ \\
			$b_2$                  & inter-species competition in $B$   \\
			$a_{MB}$               & rate at which $M$ are recruited into $B$ \\
			$\frac{1}{d_1}$        & Measures the effect of local ecological factors in promoting the growth\\ &  rate of $M$\\
			$d_2$                  & Measures the resilience of $M$ in recruitment to $B$ \\
			\hline
		\end{tabular}
	\end{center}
\end{table}
The term $\boxed{\frac{a_1M}{1+d_1 B}}$ represents the growth rate of the Monolingual population, which could be hindered by the Bilingual population. Local ecological factors will influence $d_{1}$. If local ecological factors promote Monolingual Hindi, then $d_{1} \ll 1$, and the growth of $B$ is unable to curb the growth of $M$.
The growth of the Bilingual population depends on the successful recruitment from the Monolingual population, described via $\pm \frac{a_{MB}MB}{1+d_2M}$. Notice, that if $d_2 \gg 1$ then the recruitment by $B$ is small. Hence, the growth of $B$ is reduced. For large Monolingual populations there is a maximum recruitment $a_{MB}/d_2$ by the Bilingual population. The Bilingual population experiences a loss due to inter-species competition, this may be a consequence of limitation of resources, expressed via  $\boxed{- b_2 B^2}$.

\section{Data Fitting Results}
\subsection{Fitting Via Our New Model}
We now fit the Indian census data to the solutions of the system of ODE given in Eq. ~\eqref{ODE:M}. The populations of the Monolingual Hindi speakers and English-Hindi Bilinguals from 1961-2011 are shown in Table \ref{table:census data}.

\begin{table}[b]
	\caption{Census data of English-Hindi Bilinguals and Monolingual Hindi populations of India during 1961-2011 \cite{PC16,census20111,census20112}.}\label{tab:census data}
	\centering
	\begin{tabular}{@{\hspace{5mm}}l @{\hspace{15mm}}r @{\hspace{15mm}}r}\hline
		Year &Bilingual population & Monolingual population\\ \hline
		1961   &\num{3314534}              & \num{130120826}     \\
		1971   &\num{8500000 }             & \num{194267971  }   \\
		1981   &\num{17000000  }           & \num{240749009  }    \\
		1991   &\num{29000000  }           & \num{300505193 }     \\
		2001   &\num{32371131  }           & \num{389677511}      \\
		2011   &\num{32017840}             & \num{496329353}  \\  \hline
	\end{tabular}
\label{table:census data}
\end{table}

The goodness of fit in the fitted solutions is measured using the Squared Sum of Errors (SSE) \cite{merriman1909text}. The SSE value represents the error between the original census data and the fitted values, for both the Monolingual and Bilingual data. The smaller the SSE value, the better the fit is. The best fit parameters are shown in Table \ref{para estimates-Mono_noCompetition}. For the best fit curves plotted against the census data and the SSE value, see Fig. \ref{fig:Mono_noCompetition_bestFit}.


\begin{figure}[H]
	\centering
	\includegraphics[width=.6\linewidth]{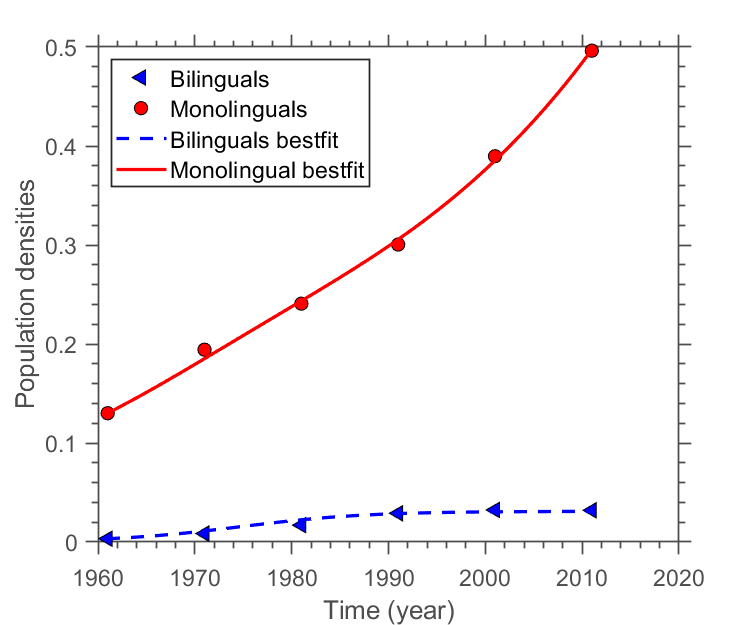}
	\caption{The trends of the Monolingual and Bilingual population densities from 1961-2011. The symbols represent the census numbers in Table \ref{para estimates-Mono_noCompetition}. Lines represent the best fits obtained from the model equation \eqref{ODE:M}. The values on the y-axis are scaled by $1\times 10^{-9}$. The resulting SSE value of the fits was $1.77\times 10^{-4}$.}
	\label{fig:Mono_noCompetition_bestFit}
\end{figure}
\begin{table}[H]
	\caption{The best estimated parameters of the model equation \eqref{ODE:M}. The scaled estimates (shown with *) are correspond to the best fits shown in Fig. \ref{fig:Mono_noCompetition_bestFit} whereas the un-scaled estimates correspond to the original census data in Table \ref{tab:census data}. }
	\centering
	\label{para estimates-Mono_noCompetition}
	\begin{tabular}{@{\hspace{5mm}}l @{\hspace{15mm}}r @{\hspace{15mm}}l}
		\hline
		Parameter& Estimate (scaled) & Estimate (un-scaled)\\
		\hline
		$a_1^*$    & \num{1.295} & $a_1^*$  \\
		$a_{MB}^*$ & \num{1013.749} &$a_{MB}^*/10^{9}$ \\
		$d_1^*$    & \num{0.171} & $d_1^*/10^{9}$ \\
		$d_2^*$    & \num{6565.040} &$d_2^*/10^{9}$ \\
		$b_1^*$    & \num{1.252} & $b_1^*$\\
		$b_2^*$    & \num{4.976} &$b_2^*/10^{9}$\\
		\hline
	\end{tabular}
\end{table}


The long-term simulations of the model solutions with the best-estimated parameters in Table \ref{para estimates-Mono_noCompetition} are shown in Fig. \ref{fig:Mono_noCompetition_bestFit_LR}.


\begin{figure}[H]
	\centering
	\subfigure{\includegraphics[width=.5\linewidth]{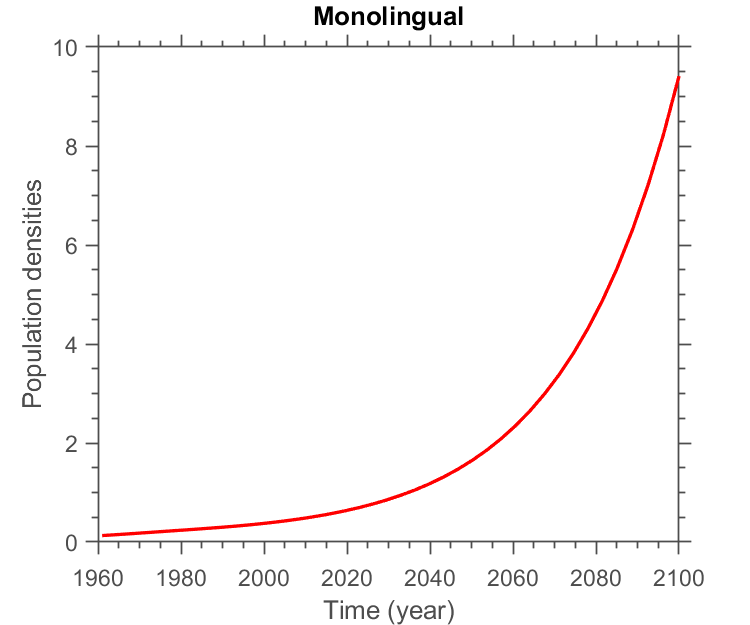}}\label{fig:Mono_noCompetition_bestFit_LR_mono}%
\subfigure{\includegraphics[width=.5\linewidth]{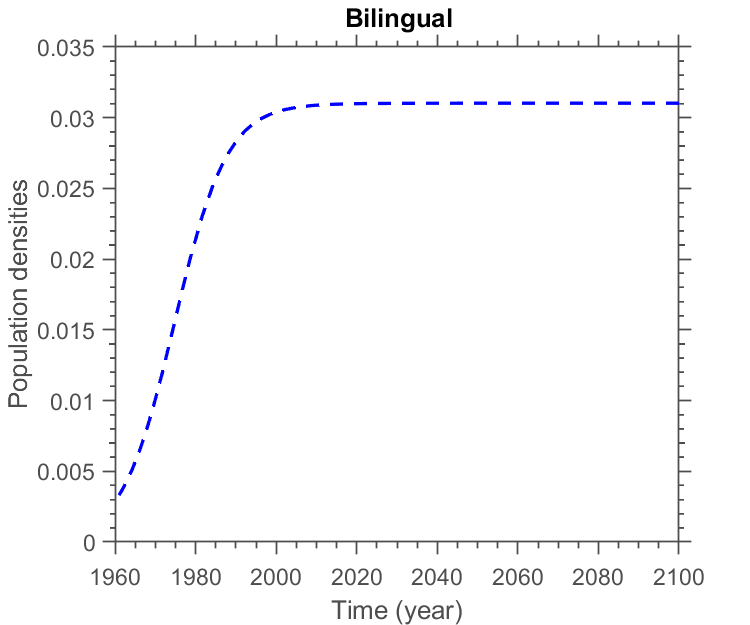}}\label{fig:Mono_noCompetition_bestFit_LR_bi}
		\caption{The long-time simulations until 2100 of the best fits in Fig. \ref{fig:Mono_noCompetition_bestFit} for Monolinguals (left panel) and Bilinguals (right panel).}
	\label{fig:Mono_noCompetition_bestFit_LR}
\end{figure}
The results of our model are compared against other ODE models in the literature and they are given in the next section.
\subsection{Fitting via the Language Dynamics Models in the Literature}
In this section, we compare our results with several other language dynamics models in the literature - in particular the Parshad \& Chand model \cite{PC16}, the Isern $\&$ Fort model \cite{IF14}, Kandler's language shift model \cite{AK10} (modified according to the context studied in this paper), the model of Mira $\&$ Paredes \cite{M05} (which reduced to the Abrams $\&$ Strogatz model after modifying according to the context of this paper). In the Abrams $\&$ Strogatz, the model has described for population fractions instead of densities. Population fractions were computed such that $M=n_M/(n_M+n_B)$ and in similar fashion for $B$ where $n_i$ represents the population density of group i. The best-fit parameters of all these model comparisons with their SSE values are given in Table \ref{all comparison models}. These models were fitted for the Indian census data and fits are shown in Fig. \ref{fig:LitModels_ALL_FITS}. We use these same models in two periods 1961-1991 and 1991-2011 to show that the data are well explained by the literature models from 1961-1991 when Bilingual population increase exponentially, see Table \ref{table: smt}.
\begin{table}[H]
	\centering
	\caption{The parameter estimates of other language dynamics models in the literature. The value for $K$, the total population is used as $0.4\times 10^{9}$. The parameters were obtained for the data scaled by a factor of $1\times 10^{-9}$. The * is used to highlight that the data used were population fractions.}
\begin{tabular}{l|l|l}
	Model                       & Parameter Estimates                                                                             & SSE                   \\ \hline
	Parshad \& Chand (2016)    & $\epsilon=2.759,d_2=37.448$                                                                     & $5.87\times 10^{-2}$  \\
	Isern \& Fort (2014)        & \begin{tabular}[c]{@{}l@{}}$\gamma=0.479,\alpha=17.404$\\ $\beta = 0.828, a=0.050$\end{tabular} & $1.43 \times 10^{-2}$ \\
	Kandler's Language Shift Model (2010) & \begin{tabular}[c]{@{}l@{}}$a_1 = 0.677, a_2=0.722$\\ $c_{12}=0$\end{tabular}                   & $6.50\times 10^{-2}$ \\
		Abrams \& Strogatz Model (2003)* & \begin{tabular}[c]{@{}l@{}}$c =13.154 , S_M=0.179$\\ $a=2.16$\end{tabular}                   & $4.68\times 10^{-2}$
\end{tabular}
\label{all comparison models}
\end{table}

\begin{figure}[H]
	\centering
	\includegraphics[width=.5\linewidth]{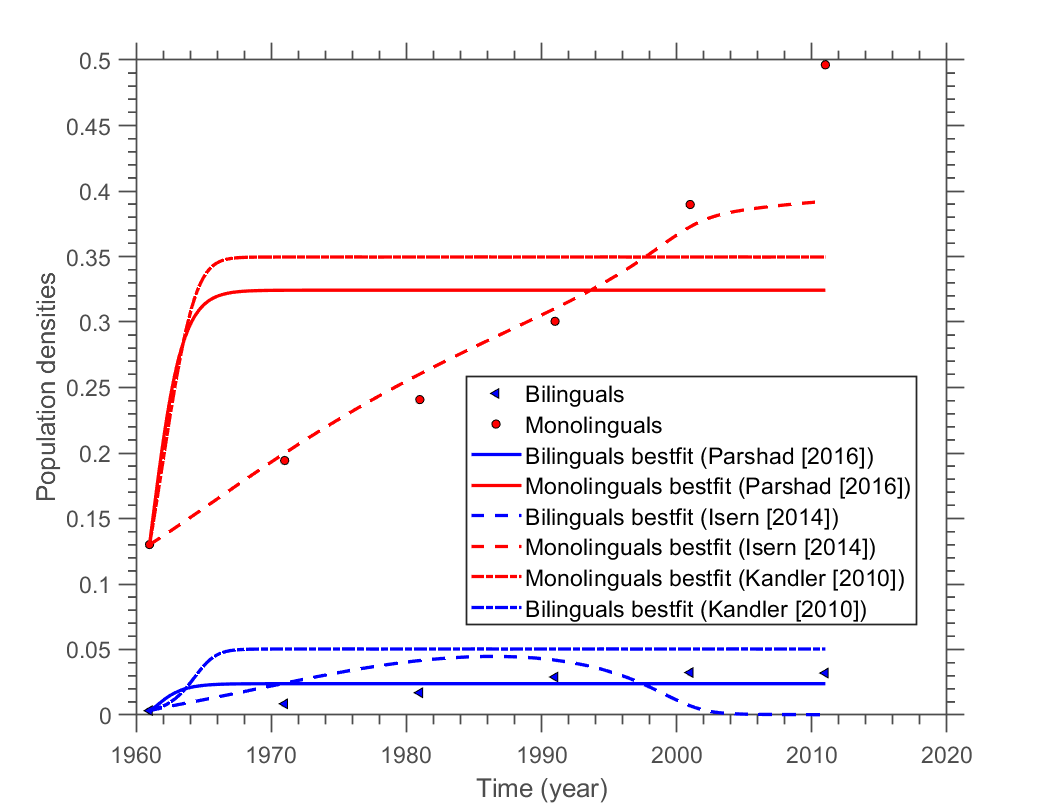}
	\caption{The trends of the Monolingual and Bilingual population densities over time compared for the literature models. The symbols represent the census data from 1961-2011. The solid lines, dash lines and dash-dot lines represent the best fits obtained from Parshad \& Chand model (Eq. (2) in \cite{PC16}), Isern \& Fort model in \cite{IF14} and Kandler's language shift model in \cite{AK10} respectively. }
	\label{fig:LitModels_ALL_FITS}
\end{figure}

\begin{table}[]
		\centering
	\caption{The SSE value comparison across models from 1961-1991 and 1991-2011.}
	\begin{tabular}{l|l|l|l|l}
		\hline
		& \multicolumn{4}{c}{SSE}                                                                                                                                                                                    \\ \cline{2-5}
		Year range  & Our Model (2019) & \begin{tabular}[c]{@{}l@{}}Parshad \& Chand \\ Model (2016)\end{tabular} & \begin{tabular}[c]{@{}l@{}}Isern \& \\ Fort (2014)\end{tabular} & \begin{tabular}[c]{@{}l@{}}Kandler \\ (2010)\end{tabular} \\ \hline
		1961 - 1991 & 5.06e-4    & 5.86e-3                                                          & 3.4e-4                                                          & 3.3e-2                                                    \\
		1991 - 2011 & 7.12e-8    & 9.4e-3                                                           & 1.05e-2                                                         & 1.6e-2
	\end{tabular}
	\label{table: smt}
\end{table}
\section{Dynamical Analysis}

\subsection{Equilibrium Solutions}

The solutions to Eq. ~\eqref{ODE:M} are nonnegative and remained bounded for all time.  They can grow at most exponentially. This is encapsulated in the following lemma.

\begin{lemma}
	\label{bound}
	For non negative initial conditions, the solutions of the system in Eq.~\eqref{ODE:M} are always positive and bounded, for any finite time $T^{*} < \infty$.
\end{lemma}
\begin{proof}
Suppose there exists a $t=\alpha$ such that $M(\alpha)<0$ or $B(\alpha)<0$.  Without loss of generality, suppose $M(\alpha)<0$ then there exists a $t^*$ such that $0\leq t^*< \alpha$ and $M(t^*)=0$.  At $t=t^*$, $\dot{M}(t^*) = 0$ hence $M(t)=0$ for $t\geq t^*$ which contradicts our assumption that $M(\alpha)<0$.  Hence, the solutions remain positive.
	
For boundedness, one sees that the state variable $B$ is bounded via comparison to the logistic equation. Next using positivity of solutions and parameters, from Eq.~\eqref{ODE:M} we see that $\frac{dM}{dt} \leq (a_{1}-b_{1})M$, and thus
$M \leq M_{0}e^{(a_{1}-b_{1})t}$, and can grow at most exponentially.
\end{proof}

Consider the solutions to the steady state equations
\begin{eqnarray}
	\label{Steady:N}
	f_{M}(M,B)M &=&\left( \frac{a_{1}}{1+d_{1}B}-\frac{a_{MB}B}{1+d_{2}M}-b_{1}\right) M=0  \\
	\label{Steady:R}
	f_{B}(M,B)B &=&\left( \frac{a_{MB}M}{1+d_{2}M}-b_{2}B\right) B=0.
\end{eqnarray}
\noindent If $B=0$ then the mathematical model becomes a harvesting equation with equilibrium solutions $(0,0)$ and $(a_{1}/b_{1},0)$. Next, consider the case where both populations are nonzero.  Clearly,
\begin{equation*}
\frac{a_{MB}M}{1+d_{2}M}-b_{2}B=0.
\end{equation*}
\noindent Solving for $B$ yields,
\begin{equation*}
B^{\ast }=\frac{a_{MB}M}{b_{2}(1+d_{2}M)}.
\end{equation*}%
\noindent If $M\neq 0$ then from Eq. (\ref{Steady:N}) one has
\[ \frac{a_{1}}{1+d_{1}B^{\ast }}-\frac{a_{MB}B^{\ast }}{1+d_{2}M}-b_{1}=0. \]
\noindent Upon solving for $M$ and substituting in for $B^\ast$ generates the cubic equation, namely,
\begin{equation*}
aM^{3}+bM^{2}+cM+d=0
\end{equation*}%
where
\begin{eqnarray*}
	a &=&-b_{2}d_{2}^{2}\left(a_{MB}b_{1}d_{1}-a_{1}b_{2}d_{2}+b_{1}b_{2}d_{2}\right) \\
	b &=&-\left(a_{MB}^{3}d_{1}+a_{MB}^{2}b_{2}d_{2}-3a_{1}b_{2}^{2}d_{2}^{2}+3b_{1}b_{2}^{2}d_{2}^{2}+2a_{MB}b_{1}b_{2}d_{1}d_{2}\right)\\
	c &=&-b_{2}\left(a_{MB}^{2}+b_{1}d_{1}a_{MB}-3a_{1}b_{2}d_{2}+3b_{1}b_{2}d_{2}\right)\\
	d &=&\allowbreak b_{2}^{2}\left( a_{1}-b_{1}\right). \\
\end{eqnarray*}

\begin{lemma}
	Let $p(x)=a_0x^{b_0} + a_1x^{b_1}+\cdots+a_nx^{b_n}$ be a polynomial with nonzero real coefficients $a_i$, where the $b_i$ are integers satisfying $0\leq b_0< b_1<b_2<\cdots <b_n$. If $a_0a_n>0$, then $z(p)$, the number of positive zeros of $p$ counting multiplicities is even; if $a_0a_n<0$ then $z(p)$ is odd.
\end{lemma}

\begin{lemma}
\label{lem:l1}
	Consider Eq. ~\eqref{ODE:M}, there are at most 2 positive interior equilibrium solutions, as long as $a_{1} < b_{1}$.
\end{lemma}
\begin{proof}
For our system, $a_0a_3 = ad =  (a_1+b_1)^2 b_2d_2 - a_{MB}b_1d_1(a_1-b_1) > 0$, if $a_{1} < b_{1}$. Hence there are an even number of positive roots which have to be 2, via Lemma \ref{lem:l1}
\end{proof}

\subsection{Linear Stability Analysis}

The Jacobian of the nonlinear system is
\begin{eqnarray}
J&=&  \left[\begin{array}{ccc}
\frac{\partial f_M}{\partial M}M + f_M & & \frac{\partial f_M}{\partial B}M \\
& & \\
\frac{\partial f_B}{\partial M}B       & & \frac{\partial f_B}{\partial B}B + f_B
\end{array}%
\right],
\end{eqnarray}
where
\begin{eqnarray*}
	\begin{array}{lcl}
		\frac{\partial f_M}{\partial M} = \frac{d_2 a_{MB} B}{(1+d_2M)^2} &&~~~ \frac{\partial f_B}{\partial M} = \frac{a_{MB}}{(1+d_2M)^2} \\
		&&\\
		\frac{\partial f_M}{\partial B} = -\frac{a_1d_1}{(1+d_1B)^2} - \frac{a_{MB}}{1+d_2M} &&~~~ \frac{\partial f_B}{\partial B} = -b_2.
	\end{array}
\end{eqnarray*}
%
%
\subsubsection{The stability of the interior equilibrium}


\begin{thm}
	\label{thm:wsoloo}
	Consider the model described by \eqref{ODE:M}. Let $(M^\ast, B^\ast)$ be an interior equilibrium point and $J\equiv J(M^\ast,B^\ast)$.  Then $(M^\ast,B^\ast)$ is locally asymptotically stable if the following conditions are satisfied:
	\begin{enumerate}
		\item $\trace(J)=\frac{d_2 a_{MB} M^\ast}{(1+d_2M^\ast)^2}  <b_2  $
		\item $\deter(J) =  \frac{1}{d_2 B^\ast} \left(\frac{a_1d_1}{(1+d_1B^\ast)^2} + \frac{a_{MB}}{1+d_2M^\ast}\right) > b_2  $
	\end{enumerate}
\end{thm}

What we note is the equilibriums that occur with the optimal/best-fit parameters, the trace and determinant are
$\trace(J) = -0.0056 < b_2 = 4.976\times 10^{-9}, \deter(J) = -0.1179 < b_2 = 4.976\times 10^{-9}$ both negative, and hence we have an unstable (saddle). This is visible in Fig. \ref{fig:Mono_noCompetition_bestFit_LR} where the Monolingual population continues to grow.

\begin{remark}
As seen in numerical simulations, the two interior equilibrium points $E_{1}^{\ast
}=\left( M_{1}^{\ast },B_{1}^{\ast }\right) $ and $E_{2}^{\ast }=\left(
M_{2}^{\ast },B_{2}^{\ast }\right) $ collide with each other and system \eqref{ODE:M} has the unique instantaneous interior equilibrium (saddle--node interior equilibrium) $\bar{E}=\left(\bar{M}, \bar{B}\right)$. Also one of
the eigenvalues of the Jacobian evaluated at the point $\bar{E}$ becomes
non-hyperbolic and its stability cannot be studied by the linearization
technique. Thus there is a chance of bifurcation around the instantaneous
interior equilibrium.
\end{remark}

This is demonstrated via the following theorem.
\begin{thm}\label{saddle:a_MB}
	System \eqref{ODE:M} experience a saddle--node bifurcation around $\bar{E}$ at $\hat{a}%
	_{MB},$ where $\hat{a}_{MB}=\left( 1+d_{2}\bar{M}\right) \left( b_{2}d_{2}%
	\bar{B}-\frac{a_{1}d_{1}}{\left( 1+d_{1}\bar{B}\right)^2 }\right) $ if $\bar{E}
	$ exists and $\left( \frac{d_{2}\hat{a}_{MB}}{(1+d_{2}\bar{M})^{2}}\right)
	\bar{M}<b_{2}.$
\end{thm}

\begin{proof}
	\label{thm:Thm312p} According to Sotomayor's theorem one of the eigenvalues
	of the Jacobian $J_{\bar{E}}$ at the saddle--node equilibrium point $\bar{E}$
	will be zero iff det$J_{\bar{E}}=\left( j_{11}j_{22}-j_{12}j_{21}\right) =0,$
	which gives $a_{MB}=\hat{a}_{MB}.$ The other eigenvalue is basically $trace$
	$J_{\bar{E}}=\left( j_{11}+j_{22}\right) $ evaluated at $a_{MB}=\hat{a}_{MB}$
	must have negative real part to get saddle--node bifurcation
	\cite{Perko2013}, so we need to take $\left( \frac{d_{2}a_{MB}\bar{B}}{%
		(1+d_{2}\bar{M})^{2}}\right) \bar{M}-b_{2}\bar{B}<0\Rightarrow \left( \frac{%
		d_{2}a_{MB}}{(1+d_{2}\bar{M})^{2}}\right) \bar{M} < b_{2}.$
	
	Let $V$ and $W$ are the eigenvectors corresponding to eigenvalue $0$ of the
	matrix $J_{\bar{E}}$ and its transpose, respectively. We obtain that $%
	V=\left( v_{1},v_{2}\right) ^{T}$ and $W=\left( w_{1},w_{2}\right) ^{T},$
	where $v_{1}=-\frac{j_{12}v_{2}}{j_{11}}=-\frac{j_{22}v_{2}}{j_{21}},$ $%
	w_{1}=-\frac{j_{21}w_{2}}{j_{11}}=-\frac{j_{22}w_{2}}{j_{12}}$ and $%
	v_{2},w_{2}\in B-\left\{ 0\right\} .$ Now let $F=\left[ \frac{a_{1}M}{%
		1+d_{1}B}-\frac{a_{MB}MB}{1+d_{2}M}-b_{1}M,\frac{a_{MB}MB}{1+d_{2}M}%
	-b_{2}B^{2}\right] ^{T}$ and $U=\left( M,B\right) ^{T},$then $W^{T}\left[
	F_{a_{MB}}(U,\hat{a}_{MB})\right] =\left( w_{1},w_{2}\right) \left( -\frac{MB%
	}{1+d_{2}M},\frac{MB}{1+d_{2}M}\right) ^{T}=\frac{MB}{1+d_{2}M}w_{2}\left(
	\frac{j_{21}}{j_{11}}+1\right) \neq 0,$ and $W^{T}\left[ D^{2}F(U,\hat{a}%
	_{MB}\left( V,V\right) \right] \neq 0.$ So from Sotomayor's theorem the
	system undergoes a saddle--node bifurcation around the positive interior
	equilibrium $\bar{E}$ at $a_{MB}=\hat{a}_{MB}$
\end{proof}
Keeping all parameters fixed we can see the coexisting
equilibrium points \ $E_{1}^{\ast }=\left( M_{1}^{\ast },B_{1}^{\ast
}\right) $ and $E_{2}^{\ast }=\left( M_{2}^{\ast },B_{2}^{\ast }\right) $
collide with each other through saddle--node bifurcation $a_{MB}$ crosses
the critical magnitude $\hat{a}_{MB}=\left( 1+d_{2}\bar{M}\right) \left(
b_{2}d_{2}\bar{B}-\frac{a_{1}d_{1}}{\left( 1+d_{1}\bar{B}\right)^2 }\right) ,$
and then mutually annihilated. The parametric surface
\newline
$\Gamma =\left\{ \left( a_{1},a_{MB},b_{1},b_{2},d_{1},d_{2}\right) \in
R_{+}^{6}:E_{1}^{\ast }=E_{2}^{\ast }=\bar{E}\text{ }real\text{
	positve }root\right\} $
	is known as the saddle-node bifurcation surface.

\begin{thm}\label{saddle:b_2}
	The system \eqref{ODE:M} undergoes a saddle--node bifurcation around $\bar{E}$ at $\hat{b}%
	_{2},$ where $\hat{b}_2= \frac{1}{d_2 \bar{B}} \left(\frac{a_1d_1}{(1+d_1 \bar{B})^2} + \frac{a_{MB}}{1+d_2\bar{M}}\right)$ if $\bar{E}
	$ exists and $\left( \frac{d_{2} a_{MB}}{(1+d_{2}\bar{M})^{2}}\right)
	\bar{M}<\hat{b}_{2}.$
\end{thm}

\begin{proof}
The proof is similar to the proof in Theorem \ref{saddle:a_MB}
\end{proof}

From the best fit parameters, $\bar{E}=(0.000152389, 0.0000848151)$ and from Theorem \ref{saddle:b_2}\\
\begin{align*}
\hat{b}_2 &= \frac{1}{d_2 \bar{B}} \left(\frac{a_1d_1}{(1+d_1 \bar{B})^2} + \frac{a_{MB}}{1+d_2\bar{M}}\right) \\
& = \dfrac{1}{6565.040\times 0.0000848151}\left(\dfrac{1.295\times 0.171}{(1+0.171\times 0.0000848151)^2}+\dfrac{1013.749}{(1+6565.040\times 0.000152389} \right)  \\
& \approx 910.51
\end{align*}

and

$$ \left( \frac{d_{2} a_{MB}}{(1+d_{2}\bar{M})^{2}}\right) \bar{M}=\left(\dfrac{6565.040\times 1013.749}{(1+6565.040\times 0.000152389)^2} \right) 0.000152389 \approx 253.44<\hat{b}_{2}. $$


\subsection{Numerical Simulations}

Note that the system \eqref{ODE:M} experience a saddle-node bifurcation around
\newline
$\bar{E}=(0.000152389 ,0.0000848151)$  at $\hat{b}_2=910.510098046$, see Figure \ref{fig:null_b2}. We note here that the system \eqref{ODE:M} does not experience saddle-node bifurcation with the optimal parameters at $\hat{a}_{MB}$ since the two interior equilibrium points close to the point of collision are both unstable. This is also true for $d_{1}, d_{2}$.

\begin{figure}[H]
	\centering
\subfigure{\includegraphics[width=.3\linewidth]{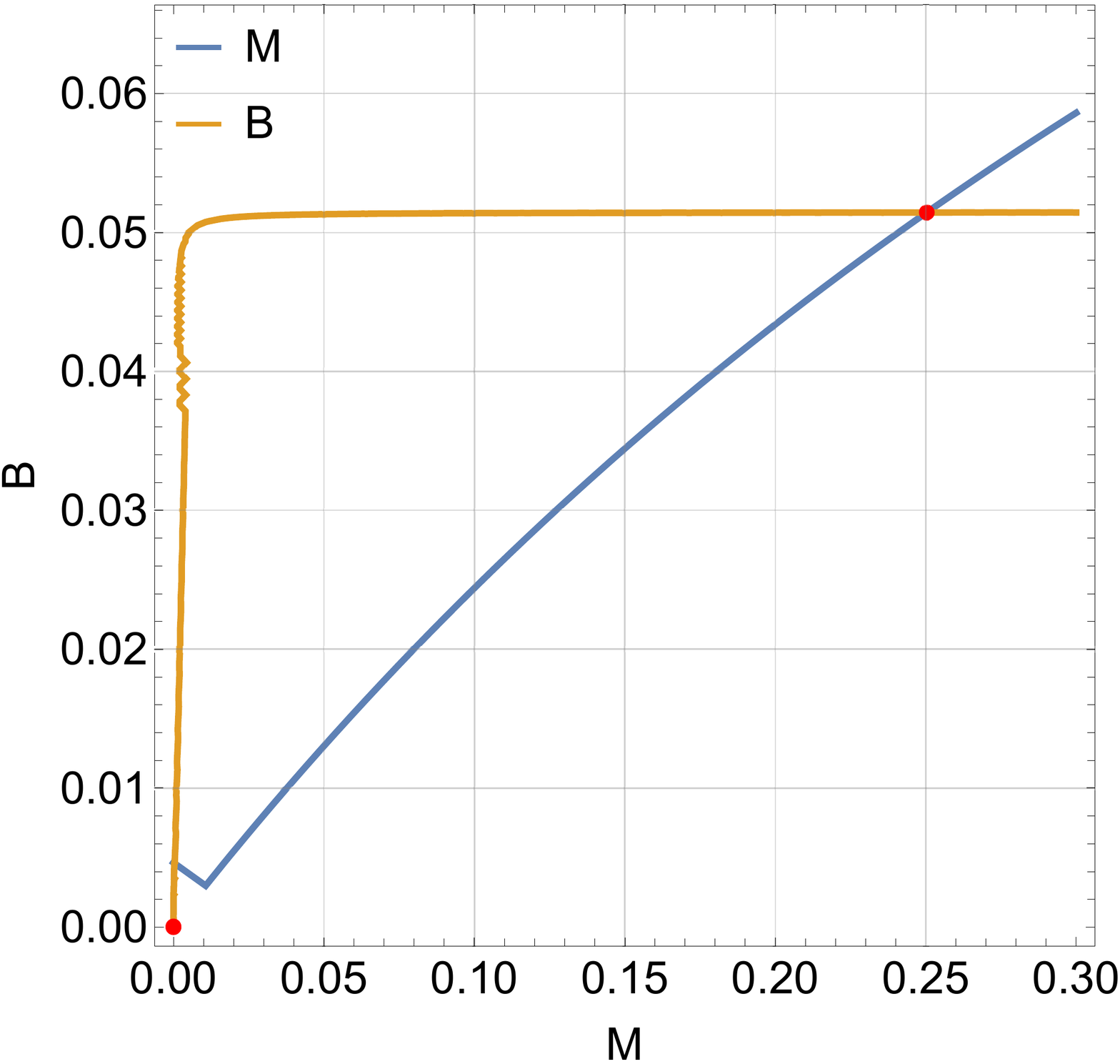}}%
\subfigure{	\includegraphics[width=.3\linewidth]{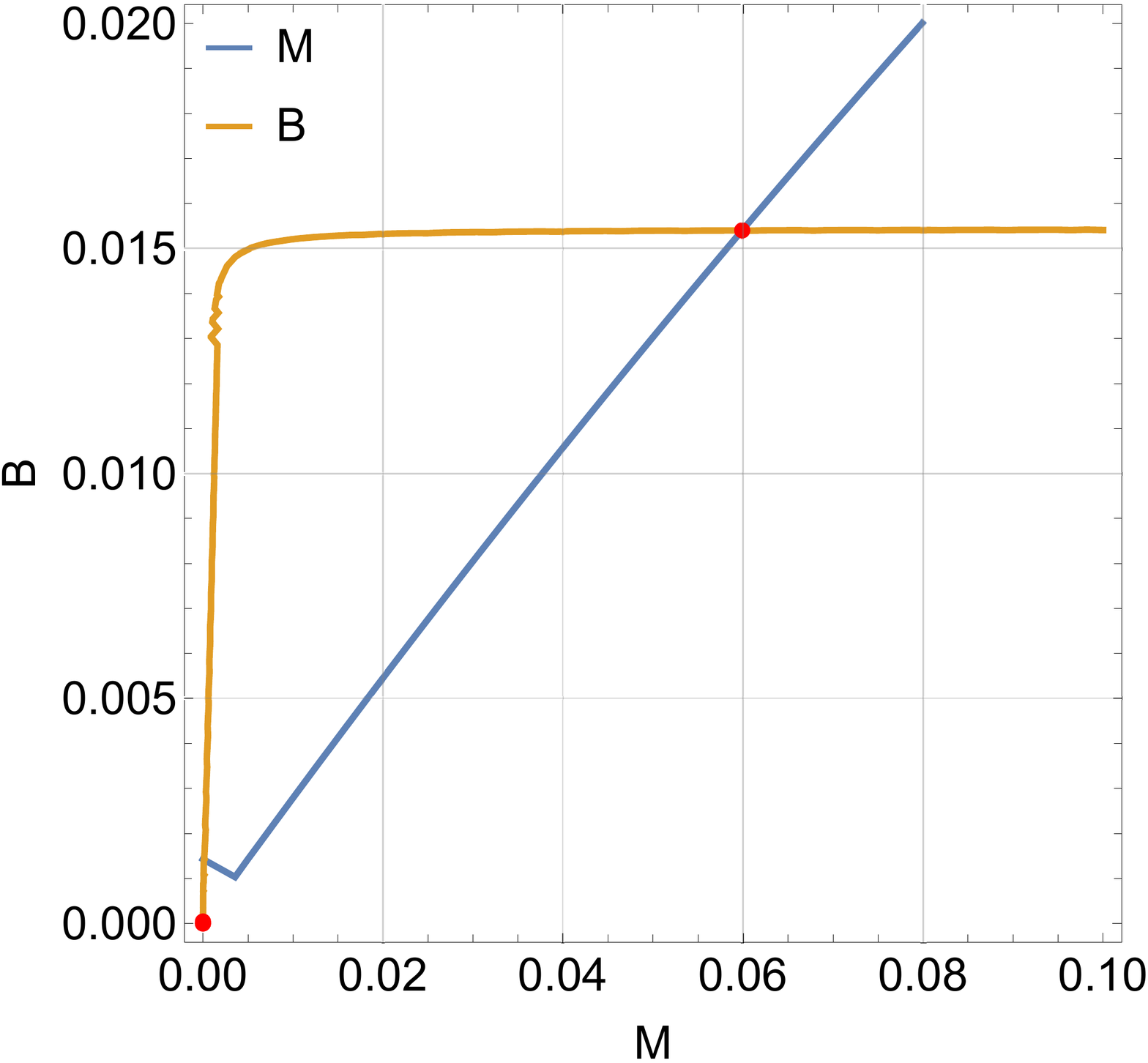}}
\subfigure{\includegraphics[width=.3\linewidth]{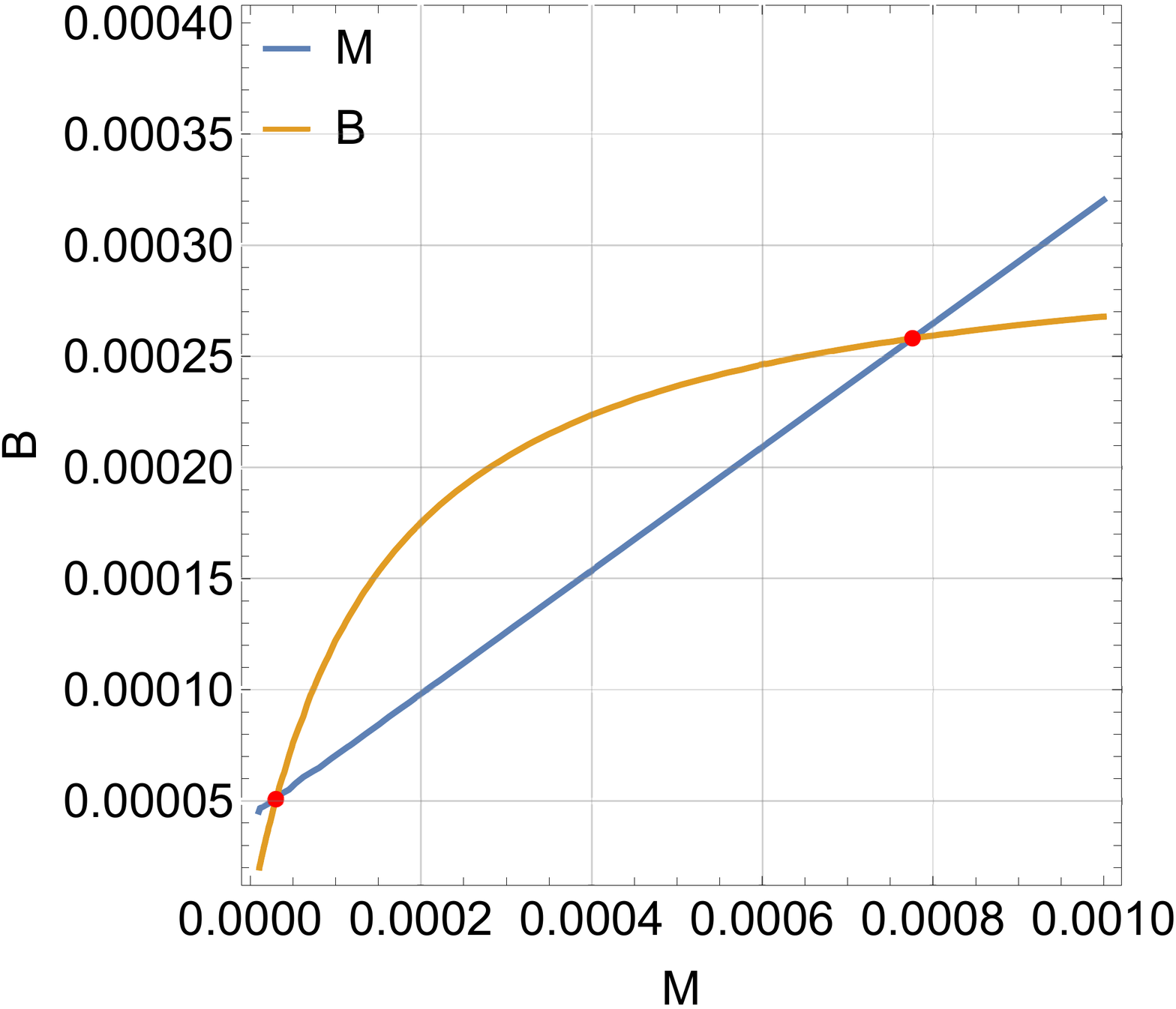}}%
\subfigure{\includegraphics[width=.3\linewidth]{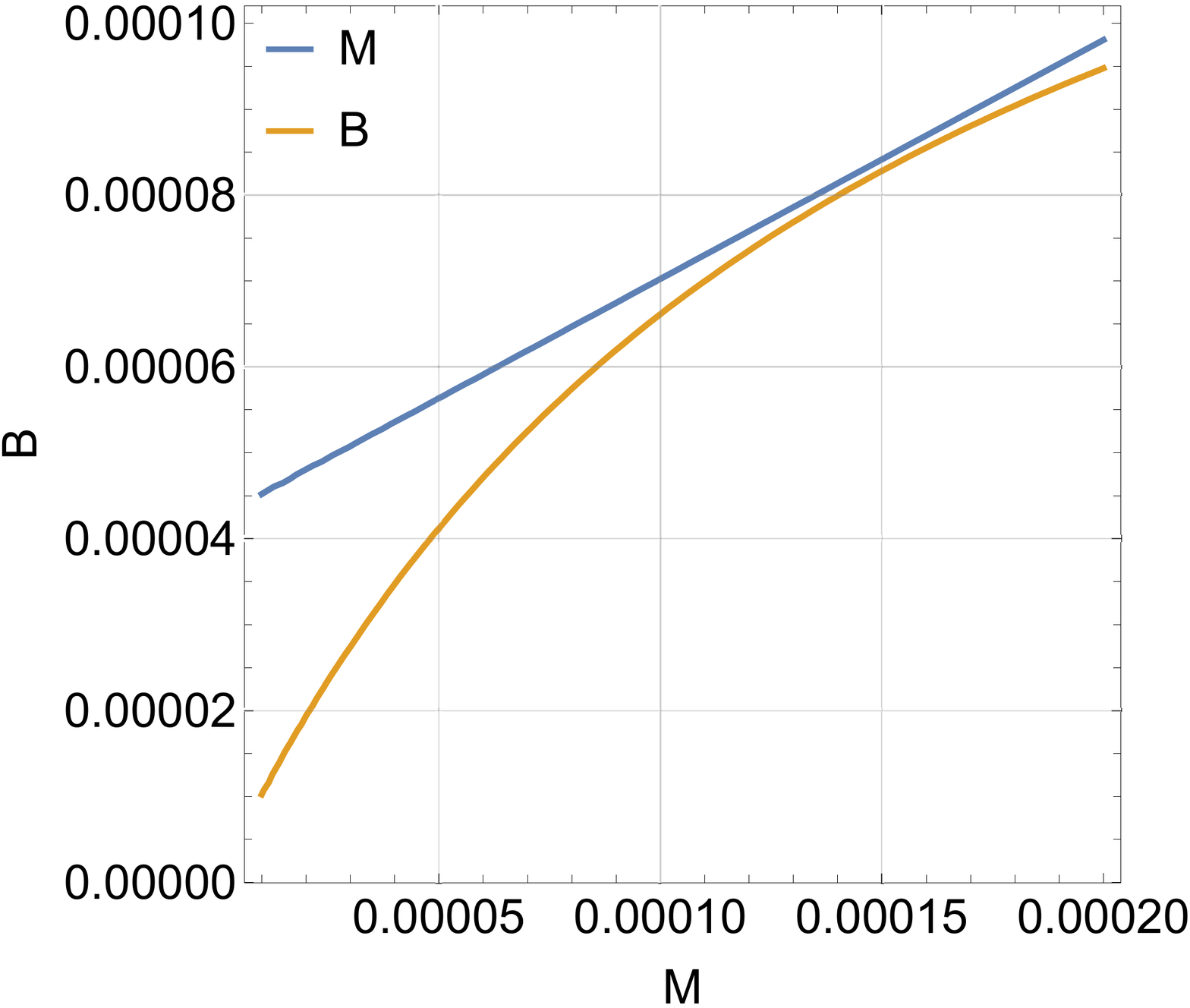}}
	\caption{The graphical illustration of nullclines and existence of interior equilibria showing the route to a saddle node bifurcation. (A) $b_2=3$, (B) $b_2=10$, (C) $b_2=500$, (D) $b_2=925$. }
	\label{fig:null_b2}
\end{figure}

\section{Discussion and Conclusion}

Language dynamics models have been useful in predicting population based language shift towards English, with great success. These models have accurately captured census data of many countries, where English is considered a high status language - and as per classical theories of language competition, the population has shifted towards speaking English (the higher status language), becoming Bilinguals, and then subsequent generations have seen to be speaking Monolingual English, such as in Wales or Scotland. A similar trend is observed in India, where English is clearly a higher status language than Hindi, due to the merits it possesses in terms of economic empowerment. In fact the 1961-1991 data shows the rise of English at a faster rate than Hindi, and classical language dynamics models that adhere to this principle, do in fact predict the best fit to this data, see Table \ref{table: smt}. However, if one looks at the data from 1991-2011, a drastic change is seen. English speakers have saturated while monolingual Hindi continues to grow exponentially.

In order to accurately predict the overall trend in the Indian census data form 1961-2011, we propose a new model that takes into account (1) local ecological factors that are promoting Hindi, such as the \emph{hindutva} ideology of the ruling BJP party, as well as (2) Competition as a saturating factor on the Bilingual population, due to lack of enough resources for a full fledged English education. Our model captures the Indian census data with near perfect accuracy (SSE 1.77 $\times$ $10^{-4}$), out performing all of the other well known models from the language dynamics literature, see Table \ref{all comparison models}. Interestingly, the best fit parameters according to our model, predict an equilibrium which is unstable - hence continued growth of Monolingual Hindi speakers. There is a second stable equilibrium in the phase space, but it is at extremely low density and unrealistic for the numbers of the Indian population. From a phase analysis
point of view we vary the parameters, $d_{1},d_{2},a_{MB},b_{2}$ in order to see how the equilibrium populations change.
A saddle node bifurcation occurs only in the case of varying $b_{2}$, see Figure \ref{fig:null_b2}, however at collision and the creation of a stable equilibrium, is seen to occur at very low density. It would be of interest to consider possible means of getting the Monolingual population to saturate, without manipulating the birth rate - this would of course be an obvious choice. It might also make for further interesting analysis if one increases the resources for Bilinguals, or looks at other means of control such as in Theoretical Ecology \cite{ParBioSci, PWB19}.

All in all we provide a first example of a high status language that has been out competed by a low status language, due to the presence of ecological factors that could curb the increase of the high status language, whilst promoting the low status language. This provides a gateway to investigating language shift with a similar social-political setting, in other centers around the world.

\pagebreak
\appendix
\section{Language Dynamics Models}
\subsection{Parshad \& Chand Model (2016)} \label{appA}
The model in \cite{PC16},
\begin{align}
\dfrac{dM}{dt} = M\left[ \dfrac{-\epsilon B}{B+M} + 1 -\dfrac{M}{K} \right]\notag\\
\dfrac{dB}{dt} = B\left[ \dfrac{\epsilon M}{B+M} - \dfrac{d_2 B}{B+M}  \right]
\end{align}
where $M(t)$ and $B(t)$ are state variables representing the populations of Monolinguals (strictly Hindi) and Bilinguals (Hindi bilingual with English). Further $\epsilon$ and $d_2$ are model parameters and $K$ is the total population.

\subsection{Isern \& Fort Model (2014)} \label{app-Isern}
The Isern's language model in \cite{IF14} describes the rate of change in the population densities of two linguistic groups M and B,

\begin{eqnarray}\label{Isern ODE}
\dfrac{\partial n_M}{\partial t}=a n_M\left(1- \dfrac{n_M+n_B}{K}\right)+\dfrac{\gamma}{\left(n_M + n_B\right)^{\alpha+\beta-1}n_M^{\alpha}n_B^{\beta}} \notag\\
\dfrac{\partial n_B}{\partial t}=a n_B\left(1- \dfrac{n_M+n_B}{K}\right)-\dfrac{\gamma}{\left(n_M + n_B\right)^{\alpha+\beta-1}n_M^{\alpha}n_B^{\beta}}
\end{eqnarray}
where $n_M$ and $n_B$ are population densities of languages groups M and B respectively and K is the carrying capacity. Further $\gamma$ is a time-scaling parameter and $\alpha,\beta\geq1$ are two parameters related to the attraction of both languages M and B.

\subsection{Mira $\&$ Paredes Model (2005)}\label{app=mira2005}
The three language competition model in \cite{M05} is given by,
\begin{align}
\dfrac{dM}{dt} &= yP_{YM} + BP_{YM}-M\left(P_{MY}+P_{MB}\right)\\
\dfrac{dB}{dt} &= MP_{MB} + BP_{YB}-B\left(P_{BM}+P_{BY}\right).\notag
\end{align}
with $M,y$ and $B$ are the fractions of populations.
The new set of equations when there is no Y monolingual group ($y=0$) after substituting to transition probabilities yields the following.
\begin{align}
\dfrac{dM}{dt} &= c\left[BS_M(1-B)^a-M(1-S_M)(1-M)^a\right]\\
\dfrac{dB}{dt} &= c\left[M(1-S_M)(1-M)^a-BS_M(1-B)^a\right] \notag
\end{align}
where the transition probabilities are $P_{BM} = cS_M(1-B)^a$ and $P_{MB} = c(1-S_M)(1-M)^a$. Moreover, we know $M+B=1$.
\subsubsection{Abrams $\&$ Strogatz model (2003)} \label{app: AK10}
Therefore we have,
\begin{align}\label{mera1-strogatz}
\dfrac{dM}{dt} &= c\left[S_M M^aB-(1-S_M)MB^a\right]\\
\dfrac{dB}{dt} &= -c\left[S_M M^aB - (1-S_M)MB^a\right] \notag
\end{align}
The system in Eq.~\eqref{mera1-strogatz} is similar to the model of Abrams-Strogatz \cite{Stro03} with $M$ and $B$ representing the \textit{monolingual} groups of languages $M$ and $B$.
If we assume there is no transition from Bilingual group to Monolingual group, then we have $P_{BM}=0$. Then the equations reduced to,
\begin{align} \label{mira2-isern}
\dfrac{dM}{dt} &= -c(1-S_M)MB^a\\
\dfrac{dB}{dt} &= c(1-S_M)MB^a \notag
\end{align}
The system in Eqs ~\eqref{mira2-isern} is similar to Isern \& Fort model in \cite{IF14}.

\subsection{Kandler's Language Shift Model (2010)} \label{app: AK10}
The dynamics of language shift model in \cite{AK10} without reaction-diffusion term is as below.
\begin{align}
\dfrac{du_1}{dt} &= a_1u_1 \left(1-\dfrac{u_1}{K-(u_2-u_3)}\right) - c_{31}u_3u_1 + c_{12}u_2u_1\\
\dfrac{du_2}{dt} &= a_2u_2 \left(1-\dfrac{u_2}{K-(u_1-u_3)}\right) + \left(c_{13}+c_{31}\right)u_1u_3 - \left(c_{12}u_1 + c_{32}u_3\right)u_2\notag\\
\dfrac{du_3}{dt} &= a_3u_3\left(1-\dfrac{u_3}{K-(u_1-u_2)}\right) - c_{13}u_1u_3 + c_{32}u_2u_3 \notag
\end{align}
This model is modified  to suit to our data by removing one of the monolingual variable, i.e. setting $u_3=0$.
\begin{align}
\dfrac{du_1}{dt} &= a_1u_1 \left(1-\dfrac{u_1}{K-(u_2)}\right)  -c_{12}u_2u_1\\
\dfrac{du_2}{dt} &= a_2u_2 \left(1-\dfrac{u_2}{K-(u_1)}\right) + c_{12}u_2u_1\notag
\end{align}
Here $u1,u2$ represent the monolingual and bilingual group respectively.
\end{document}